\documentclass[a4paper]{article}
\usepackage{url, graphicx, color, varioref, amsfonts, longtable, float}
\usepackage{epstopdf,graphicx}
\usepackage{algorithm}
\usepackage{algorithmic}

\newtheorem{theorem}{Theorem}[section]
\newtheorem{lemma}[theorem]{Lemma}

\newtheorem{definition}[theorem]{Definition}

\newenvironment{proof}[1][Proof]{\begin{trivlist}
\item[\hskip \labelsep {\bfseries #1}]}{\end{trivlist}}

\newcommand{\qed}{\nobreak \ifvmode \relax \else
      \ifdim\lastskip<1.5em \hskip-\lastskip
      \hskip1.5em plus0em minus0.5em \fi \nobreak
      \vrule height0.75em width0.5em depth0.25em\fi}

\title{A Polynomial Time Bounded-error Quantum Algorithm for Boolean Satisfiability}

\usepackage[T1]{fontenc}
\usepackage[utf8]{inputenc}
\usepackage{authblk}

\author[1,2]{Ahmed Younes\thanks{ayounes@alexu.edu.eg}}
\author[2]{Jonathan E. Rowe\thanks{J.E.Rowe@cs.bham.ac.uk}}
\affil[1]{Department of Mathematics and Computer Science, Faculty of Science, Alexandria University, Egypt}
\affil[2]{School of Computer Science, University of Birmingham, Birmingham, B15 2TT, United Kingdom}
\renewcommand\Authands{ and }

\begin{document}
\maketitle

\begin{abstract}

The aim of the paper is to answer a long-standing open problem on the relationship between NP and BQP. The paper shows 
that BQP contains NP by proposing a BQP quantum algorithm for the MAX-E3-SAT problem which is a fundamental NP-hard 
problem. Given an E3-CNF Boolean formula, the aim of the MAX-E3-SAT problem is to find the variable assignment that 
maximizes the number of satisfied clauses. The proposed algorithm runs in $O(m^2)$ for an E3-CNF Boolean formula with $m$ clauses 
and in the worst case runs in $O(n^6)$ for an E3-CNF Boolean formula with $n$ inputs. The proposed algorithm maximizes the set of 
satisfied clauses using a novel iterative partial negation and partial measurement 
technique. The algorithm is shown to achieve an arbitrary high probability of success of 
$1-\epsilon$ for small $\epsilon>0$ using a polynomial resources. 
In addition to solving the MAX-E3-SAT problem, the proposed algorithm can also be used to 
decide if an E3-CNF Boolean formula is satisfiable or not, which is an NP-complete problem, based on the maximum number of satisfied 
clauses.
  
%

\end{abstract}

\section{Introduction}  

A long-standing open problem in quantum computing is the relationship between the classes NP and BQP \cite{Bennett:1997, Buhrman:1999}. 
Decision problems are in NP if 
yes-instances have witnesses that can be checked in polynomial time \cite{Williamson:2011}. The class BQP is the quantum 
computing analogue of the classical class BPP (bounded error probabilistic polynomial) \cite{Fortnow:1999}. 
A problem is in BPP if there is a probabilistic classic algorithm 
(Turing machine with access to random bits) which makes errors (for either yes or no instances) with probability of given a wrong 
answer at most 1/3. The value 1/3 is arbitrary - all that is required is that the value is bounded away from 1/2.  By repeated runs, 
the probability of failure can be made exponentially small. The problem class BQP replaces the classical algorithm with a quantum 
algorithm \cite{Adleman:1997}. Thus a decision problem is in BQP if there is a quantum algorithm for it with probability of being wrong less than 1/3. 

The common belief concerning the relationship between NP and BQP was that NP is not contained in BQP (e.g see chapter 15 of 
\cite{moore:martens}). However, a recent paper by one of the current authors has shown that an NP-hard problem (Graph Bisection) 
can be efficiently solved, with low failure probability, by a quantum algorithm \cite{MaxMinBis}. This implies that NP is 
in fact contained within BQP, as any NP problem can be polynomially reduced to an NP-hard problem. The Graph Bisection problem is, 
perhaps, somewhat obscure and much of the presentation of that result involves ensuring the balance of the partition, which detracts 
from the main features of the approach. Consequently in the current paper, we will directly address the classic Boolean 
Satisfiability problem (SAT) to show precisely how constraints, expressed as Boolean formula, can be encoded into quantum algorithm. 
The constraints are entangled with the superposition of all possible truth-value assignments and a probability amplification 
technique applied to amplify the assignment which maximizes the number of satisfied clauses.

In particular, we will focus on MAX-E3-SAT \cite{CombOptBook}, in which each clause contains exactly three literals, and we will show that our 
quantum algorithm will solve this maximization problem with high probability of success. In particular, it can then be used to 
solve the decision problem (with high probability). Iterating the process allows the probability of failure to be made exponentially 
small.

A key fact about the MAX-E3-SAT problem is that random truth assignments will satisfy, in expectation, 7/8 of the clauses. 
A consequence of the PCP Theorem is that this cannot be improved upon (more precisely, there is no  $(7/8 + \epsilon)$ approximation 
algorithm for constant $\epsilon > 0$) unless P=NP \cite{hastad:2001}. That MAX-E3-SAT can be solved (with high probability) 
in polynomial time by a quantum computer is therefore all the more remarkable.

The result shown in this paper doesn't contradict with that shown \cite{Bennett:1997} where it was shown that 
BQP does not contain NP relative to a random unitary oracle with probability one. This argument does not imply 
that BQP does not contain NP in a non-relativized world which is the novel feature in the proposed algorithm, where 
partial measurement is used in the amplitude amplification process instead of the 
usual unitary amplitude amplification techniques that use iterative calls to an oracle to amplify the required solution.

The aim of the paper is to propose a quantum algorithm for the MAX-E3-SAT problem. The algorithm prepares a superposition 
of all possible variable assignments, then the algorithm evaluates the set of clauses using all the possible variable assignments simultaneously 
and then amplifies the amplitudes of the state(s) that achieve(s) the maximum satisfaction to the set of clauses 
using a novel amplitude amplification technique that applies an iterative partial negation and partial measurement. 
The proposed algorithm runs in $O(m^2)$ for an E3-CNF Boolean formula with $m$ clauses and in the worst case 
runs in $O(n^6)$ for an E3-CNF Boolean formula with $n$ Boolean variables to achieve an arbitrary high probability of 
success of $1-\epsilon$ for small $\epsilon>0$ using a polynomial resources.

The paper is organized as follows; Section 2 shows the data structures and the quantum circuit for encoding an E3-CNF Boolean formula. 
Section 3 presents the proposed algorithm with analysis on time and space requirements. 
Section 4 concludes the paper. 


\section{Data Structures and Clause Encoding}

An $n$ inputs $k$-CNF Boolean formula, 
\begin{equation}
f(x_0,x_1, \ldots, x_{n-1} )= c_0  \wedge c_1  \wedge  \ldots  \wedge c_{m - 1}, 
\end{equation}
 is a conjunction (AND) of $m$ clauses, 
each clause $c_j$  represented by a disjunction (OR) of exactly $k \le n$ literals, 
$c_j  = \left( {l_{j_0}  \vee l_{j_1}  \vee l_{j_2} \ldots \vee l_{j_{k-1}}} \right)$, 
such that a literal $l_{j,a}$ in clause $c_j$ with $0 \le a \le k-1$ and $0\le j \le m-1$ equals to an input 
variable in its true form $x_i$ or its complemented form ${\neg x}_i$. That is,  
$l_{j,a}= \mathop {\mathop x\limits^ \bullet  }\nolimits_i$, 
where $\mathop {\mathop x\limits^ \bullet  }\nolimits_i$ can be replaced by $x_i$ or ${\neg x}_i$ such that 
${\neg x}_i$ is the negation of $x_i$ with $0\le i \le n-1$.
The first aim is to decide whether $f$ is satisfiable or not (deciding $f$). The second aim is to 
find a variable assignment for $x_0,x_1, \ldots$ and $x_{n-1}$ that satisfies $f$ if it is satisfiable (solve $f$) or 
to find a variable assignment that satisfies the maximum possible number of clauses if $f$ is unsatisfiable (maximize $f$).

The problem of deciding whether a $k$-CNF Boolean formula is satisfiable or not is NP-complete and is known as $k$-SAT 
or E$k$-SAT problem. The optimization problem associated with the $k$-SAT problem to find a variable assignment 
to satisfy a satisfiable $k$-CNF formula is NP-hard. 
If the $k$-CNF is unsatisfiable, then the problem of finding a variable assignment to maximize the number of satisfied clauses 
is known as MAX-Ek-SAT problem which is an NP-hard problem \cite{CombOptBook}. 
The maximum number of clauses for a $k$-CNF Boolean formula is $\mathop 2\nolimits^k \left( {\begin{array}{*{20}c}
   n  \\
   k  \\
\end{array}} \right)=O(n^k)$. Without loss of generality, this paper targets the E3-SAT and MAX-E3-SAT problems where $k=3$, so 
the maximum number of clauses $m$ for the MAX-3E-SAT is ${\textstyle{4 \over 3}}n(n - 1)(n - 2)$.

\subsection{Encoding of a Solution }
A candidate solution $S$ to the MAX-E3-SAT problem is a vector of variable assignment 
$A =(x_0,x_1, \ldots ,x_{n-1}) \in \{0,1\}^n$, and each vector $A$ is associated with a vector of the truth values 
$C(A) = (c_0,c_1,\ldots,c_{m-1}) \in \{0,1\}^m$ of the $m$ clauses sorted in order, i.e. $S=(A,C(A)) \in \{0,1\}^{n+m}$. 
The optimal solution $S_{max}=(A_{max},C(A_{max}))$ is the solution that contains a vector of variable assignment $A_{max}$ 
that maximizes the number of 1's in the vector of the truth values $C(A_{max})$ of the $m$ clauses. 
For short, the number of 1's in the vector of clauses $C$, i.e. the number of satisfied clauses, will be referred to 
as the 1-density of $C$ so that the 1-density for a satisfiable formula must be equal to $m$. 
For example, consider the E3-CNF Formula with $n=3$ and $m=4$, 
\begin{equation}
 f(x_0 ,x_1 ,{\rm  }x_2 ) = c_0  \wedge c_1  \wedge c_2  \wedge c_3, 
 \end{equation}
 \noindent
 where,
\begin{equation}
\begin{array}{l}
 c_0  = (\neg x _0  \vee \neg x _1  \vee \neg x _2 ){\rm   }, \\ 
 c_1  = (\neg x _0  \vee x_1  \vee {\rm  }x_2 ), \\ 
 c_2  = (x_0  \vee \neg x _1  \vee {\rm  }x_2 ), \\ 
 c_3  = (x_0  \vee x_1  \vee {\rm  }x_2 ), \\ 
 \end{array}
\end{equation}
\noindent
then a solution to this formula will be encoded as 
$S=(A,C(A))$, where $A=({x_0 ,x_1 , x_2 })$ and $C(A)= ({c_0 ,c_1 ,c_2 ,c_3 })$. 
This formula is satisfiable when $(x_0,x_1,x_2)=(0,0,1),(0,1,1),(1,0,1)$, or $(1,1,0)$, and an instance of an 
optimal solution will be $S_{max}=((0,0,1),(1,1,1,1))$ with $A_{max}=(0,0,1)$ and $C(A_{max})=(1,1,1,1)$. 
For $n \ge 3$ and $m>7$, the E3-CNF formula might not be satisfied \cite{Dimitriou:2005} where the 1-density of the $C(A_{max})$ 
vector will give the maximum number of satisfied clauses and the order of 1's will show the satisfied 
clauses using the variable assignment $A_{max}$. 

For $n \ge 3$ and $m={\textstyle{4 \over 3}}n(n - 1)(n - 2)$, we have the problem instance compsrising all possible clauses.
The 1-density of $C(A)$ in this case will be 7m/8 which is 
the worst possible case for the 1-density of $C$. The 3-CNF formula will be 
unsatisfied for an arbitrary variable assignment $A$ \cite{{Zhang:2001}}.

\subsection{Encoding of a Clause}
\label{cfsec}
An E3-CNF formula with $n$ inputs and $m$ clauses will be encoded as an $n+m$ inputs/outputs quantum circuit.  
Every E3-CNF clause $c  = \left(l_0  \vee l_1  \vee l_2 \right)$ will be encoded using a $4 \times 4$ quantum gate. 
The $GT^4$ ( $4 \times 4$ Generalized Toffli) gate \cite{toffoli80} is the main primitive gate
that will be used to encode a clause. The $GT^{4}$ gate is defined
as follows:

\begin{definition}($GT^{4}$ gate)

$GT^{4}$ gate is a reversible gate denoted as,
\begin{equation} 
( y_0 , y_1,y_2;f_{out}) = GT^{4}(x_{0} \oplus \delta_0, x_{1} \oplus \delta_1 ,x_{2} \oplus \delta_2 ;f_{in}),
\end{equation}
\noindent
where $x_a, \delta_a, f_{in}$ and $f_{out} \in \{0,1\}$ with $a \in \{0,1,2\}$. 
The $GT^{4}$ gate has 4 inputs: $x_{0}$, $x_{1}$, $x_{2}$ (known as control
qubits) and $f_{in}$ (known as target qubit). Each control qubit $x_a$ is associated with a condition $\delta_a$, such that 
if $\delta_a=1$ then the condition on $x_a$ is satisfied if $x_a=0$, i.e. $x_a \oplus 1 = \neg x_a$, and if 
$\delta_a=0$ then the condition on $x_a$ is satisfied if $x_a=1$. The $GT^4$ gate has 4 outputs:
$y_{0}$, $y_{1}$, $y_{2}$ and ${f_{out}}$. The operation
of the $GT^{4}$ gate is defined as follows,

\begin{equation}
\begin{array}{l}
 y_a  = x_a ,\mathrm{ for}\,a=\{0,1,2\}, \\
 f_{out}  = f_{in}  \oplus \left( (x_{0} \oplus \delta_0)  \wedge  (x_{1} \oplus \delta_1) \wedge  (x_{2} \oplus \delta_{2}) \right), \\
 \end{array}
\end{equation}
\noindent
\end{definition}
\noindent
where $\oplus$ is the XOR logic operation, i.e. the target qubit $f_{in}$ will be flipped if and only if each control
qubits $x_a$ satisfies its associated condition $\delta_a$. For example, $f_{out}= \neg f_{in}$ for the gate 
$GT^{4}(x_{0} \oplus 1, x_{1}  ,x_{2} \oplus 1 ;f_{in})$ if and only if $x_0=0$, $x_1=1$ and $x_2=0$.

A $GT^4$ gate with its target qubit, $f_{in}$, initialized to state  
$\left| {1 } \right\rangle$ can be used to encode a clause $c  = \left(l_0  \vee l_1  \vee l_2 \right)$ 
using the Boolean algebraic identity,
\begin{equation}
 c=\left( {l_0  \vee l_1  \vee l_2 } \right) = \left( {\left( {l_0  \oplus 1} \right) \wedge \left( {l_1  \oplus 1} \right) \wedge \left( {l_2  \oplus 1} \right)} \right) \oplus 1, 
\label{biden}
 \end{equation}
\noindent
so that $f_{out} = c$, where $l_{a}= \mathop {\mathop x\limits^ \bullet  }$, 
and $\mathop {\mathop x\limits^ \bullet  }$ can be replaced by $x$ or ${\neg x}$  such that 
${\neg x}= x \oplus 1$ is the negation of $x$. That is,
\begin{equation}
(x_0,x_1,x_2;c)= GT^4(l_0 \oplus 1 , l_1 \oplus 1, l_1 \oplus 1 ; 1).
\end{equation}

For example, consider the following E3-CNF Boolean formula with $n=4$ and $m=3$,

\begin{equation}
f(x_0,x_1,x_2,x_3)=c_0 \wedge c_1 \wedge c_2,
\label{formula1}
\end{equation}
\noindent
with
\begin{equation}
\label{formula2}
 c_0=\left( {x_0  \vee \neg x_1  \vee \neg x_2 } \right),\\ 
 c_1=\left( {\neg x_0  \vee x_1  \vee \neg x_3 } \right), \\ 
 c_2=\left( {x_0  \vee x_2  \vee \neg x_3 } \right). \\ 
\end{equation}

Apply the Boolean algebraic identity shown in equation (\ref{biden}) on each clause, then,

\begin{equation}
\begin{array}{l}
 c_0=\left( {x_0  \vee \neg x_1  \vee \neg x_2 } \right) = \neg \left( {\neg x_0  \wedge x_1  \wedge x_2 } \right) = \left( {\neg x_0  \wedge x_1  \wedge x_2 } \right) \oplus 1, \\ 
 c_1=\left( {\neg x_0  \vee x_1  \vee \neg x_3 } \right) = \neg \left( {x_0  \wedge \neg x_1  \wedge x_3 } \right) = \left( {x_0  \wedge \neg x_1  \wedge x_3 } \right) \oplus 1, \\ 
 c_2=\left( {x_0  \vee x_2  \vee \neg x_3 } \right) = \neg \left( {\neg x_0  \wedge \neg x_2  \wedge x_3 } \right) = \left( {\neg x_0  \wedge \neg x_2  \wedge x_3 } \right) \oplus 1, \\ 
 \end{array}
\end{equation}

then each clause can be encoded using a $GT^4$ gate as follows,

\begin{equation}
\begin{array}{l}
 (x_0,x_1,x_2;c_0) \equiv  GT_0^4 \left( {x_0 \oplus 1  , x_1  , x_2 ,1 } \right),\\ 
 (x_0,x_1,x_3;c_1) \equiv  GT_1^4\left( {x_0  , x_1 \oplus 1 , x_3,1 } \right),\\ 
 (x_0,x_2,x_3;c_2) \equiv  GT_2^4\left( { x_0 \oplus 1 , x_2 \oplus 1  , x_3,1 } \right).\\ 
 \end{array}
\end{equation}

To construct a quantum circuit for this E3-CNF formula, prepare a quantum register with 4 qubits to be loaded with the values of 
$x_0$, $x_1$, $x_2$ and $x_3$, and add 3 extra qubits initialized with 
the quantum state $\left| {1 } \right\rangle$ so that $GT^4_0$ uses the first extra qubit as the target qubit,
$GT^4_1$ uses the second extra qubit as the target qubit, and so on, as shown in figure \ref{GTcircuit}. Let 
$U$ be a quantum circuit on 7 qubits defined as $U=GT_0^4 GT_1^4 GT_2^4$, then,

\begin{equation}
(x_0,x_1,x_2,x_3;c_0,c_1,c_2)=U(x_0,x_1,x_2,x_3;1,1,1).
\end{equation}

\begin{figure*}[t]
\begin{center}
\includegraphics[natwidth=245bp,natheight=175bp, width=125bp]{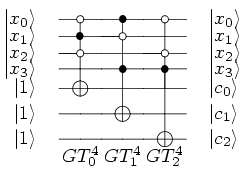}

\end{center}
\caption{A quantum circuit using $GT^4$ gates for the E3-CNF formula shown in Eqn. \ref{formula1} and \ref{formula2}, where 
$ \bullet$ on a control qubit means that the associated condition is 1 while  $\circ$ on a control qubit means that the associated condition is 0.}

\label{GTcircuit}
\end{figure*}

\section{The Algorithm}

Given an E3-CNF formula $f$ with $n$ inputs and $m$ clauses. The proposed algorithm is divided into three stages, 
the first stage prepares a superposition of all possible variable assignments for the $n$ variables. 
The second stage evaluates the $m$ clauses for every variable assignment and stores the truth values of the clauses
in truth vectors entangled with the corresponding variable assignments in the superposition.
The third stage amplifies the truth vector of clauses with maximum number of satisfied clauses using 
a partial negation and iterative measurement technique. The proposed algorithm uses $(n+m+1)$ qubits during the three stages. 
Each of the first $n$ qubits is initialized to state $\left| 0 \right\rangle$, each of the $m$ qubits 
is initialized to state $\left| 1 \right\rangle$, and an extra auxiliary qubit, denoted $\left| ax \right\rangle$ is initialized to 
state $\left| 0 \right\rangle$.
The qubit $\left| ax \right\rangle$ is an auxiliary qubit that will be used during the amplitude amplification technique. The 
amplitude of the state of $\left| ax \right\rangle$ entangled with every truth vector of clauses will act as an aggregator for the 
amount of partial negation to be applied on that state based on the 1-density of the entangled  truth vector of 
clauses. The probability of finding $\left| ax \right\rangle = \left| 1 \right\rangle$ when measured will depend of the accumulation of all partial 
negations applied on the states of that qubit.

The system is initially as follows,

\begin{equation}
\left| {\psi _0 } \right\rangle  = \left| 0 \right\rangle ^n  \otimes \left| 1 \right\rangle ^m  \otimes \left| 0 \right\rangle.   
\end{equation}

\begin{center}
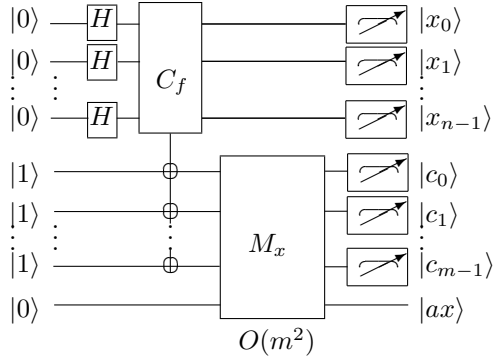
\begin{figure*}[t]
\begin{center}

\setlength{\unitlength}{3947sp}%
\begingroup\makeatletter\ifx\SetFigFont\undefined
\def\x#1#2#3#4#5#6#7\relax{\def\x{#1#2#3#4#5#6}}%
\expandafter\x\fmtname xxxxxx\relax \def\y{splain}%
\ifx\x\y   
\gdef\SetFigFont#1#2#3{%
  \ifnum #1<17\tiny\else \ifnum #1<20\small\else
  \ifnum #1<24\normalsize\else \ifnum #1<29\large\else
  \ifnum #1<34\Large\else \ifnum #1<41\LARGE\else
     \huge\fi\fi\fi\fi\fi\fi
  \csname #3\endcsname}%
\else
\gdef\SetFigFont#1#2#3{\begingroup
  \count@#1\relax \ifnum 25<\count@\count@25\fi
  \def\x{\endgroup\@setsize\SetFigFont{#2pt}}%
  \expandafter\x
    \csname \romannumeral\the\count@ pt\expandafter\endcsname
    \csname @\romannumeral\the\count@ pt\endcsname
  \csname #3\endcsname}%
\fi
\fi\endgroup
\begin{picture}(2915,2199)(583,-1812)
\thinlines
\put(2987,-1287){\oval(236, 72)[tr]}
\put(2987,-1287){\oval(236, 72)[tl]}
\put(2791,-1170){\line( 0,-1){235}}
\put(2791,-1405){\line( 1, 0){376}}
\put(3167,-1405){\line( 0, 1){235}}
\put(3167,-1170){\line(-1, 0){376}}
\put(2876,-1345){\vector( 2, 1){282}}
\put(2988,-964){\oval(236, 72)[tr]}
\put(2988,-964){\oval(236, 72)[tl]}
\put(2792,-847){\line( 0,-1){235}}
\put(2792,-1082){\line( 1, 0){376}}
\put(3168,-1082){\line( 0, 1){235}}
\put(3168,-847){\line(-1, 0){376}}
\put(2877,-1022){\vector( 2, 1){282}}
\put(2990,-695){\oval(236, 72)[tr]}
\put(2990,-695){\oval(236, 72)[tl]}
\put(2794,-578){\line( 0,-1){235}}
\put(2794,-813){\line( 1, 0){376}}
\put(3170,-813){\line( 0, 1){235}}
\put(3170,-578){\line(-1, 0){376}}
\put(2879,-753){\vector( 2, 1){282}}
\put(2982,233){\oval(236, 72)[tr]}
\put(2982,233){\oval(236, 72)[tl]}
\put(2786,350){\line( 0,-1){235}}
\put(2786,115){\line( 1, 0){376}}
\put(3162,115){\line( 0, 1){235}}
\put(3162,350){\line(-1, 0){376}}
\put(2871,175){\vector( 2, 1){282}}
\put(2979,-363){\oval(236, 72)[tr]}
\put(2979,-363){\oval(236, 72)[tl]}
\put(2783,-246){\line( 0,-1){235}}
\put(2783,-481){\line( 1, 0){376}}
\put(3159,-481){\line( 0, 1){235}}
\put(3159,-246){\line(-1, 0){376}}
\put(2868,-421){\vector( 2, 1){282}}
\put(2980,-27){\oval(236, 72)[tr]}
\put(2980,-27){\oval(236, 72)[tl]}
\put(2784, 90){\line( 0,-1){235}}
\put(2784,-145){\line( 1, 0){376}}
\put(3160,-145){\line( 0, 1){235}}
\put(3160, 90){\line(-1, 0){376}}
\put(2869,-85){\vector( 2, 1){282}}
\put(1692,-687){\oval(84,92)}
\put(1693,-937){\oval(84,92)}
\put(1694,-1276){\oval(84,92)}
\put(1692,-1015){\line( 0, 1){564}}
\put(1692,-1320){\line( 0, 1){141}}
\put(2006,-1607){\framebox(642,1012){}}
\put(2004,-690){\line(-1, 0){1034}}
\put(970,-690){\line( 0, 1){  0}}
\put(2004,-940){\line(-1, 0){1034}}
\put(970,-940){\line( 0, 1){  0}}
\put(1999,-1283){\line(-1, 0){1034}}
\put(965,-1283){\line( 0, 1){  0}}
\put(970,-1525){\line( 1, 0){1034}}
\put(949,-355){\line( 1, 0){234}}
\put(1180,-113){\framebox(176,215){}}
\put(1180,135){\framebox(176,215){}}
\put(1181,-475){\framebox(176,215){}}
\put(1359,234){\line( 1, 0){141}}
\put(1365, -2){\line( 1, 0){139}}
\put(1359,-355){\line( 1, 0){141}}
\put(945,237){\line( 1, 0){229}}
\put(1501,-445){\framebox(383,811){}}
\put(945, -2){\line( 1, 0){229}}
\put(2782,235){\line(-1, 0){893}}
\put(1889,235){\line( 0, 1){  0}}
\put(1889,235){\line( 1, 0){893}}
\put(2782,  1){\line(-1, 0){893}}
\put(1889,  1){\line( 0, 1){  0}}
\put(1889,  1){\line( 1, 0){893}}
\put(2778,-355){\line(-1, 0){889}}
\put(1889,-355){\line( 0, 1){  0}}
\put(1889,-355){\line( 1, 0){889}}
\put(2654,-687){\line( 1, 0){141}}
\put(2651,-938){\line( 1, 0){141}}
\put(2653,-1286){\line( 1, 0){141}}
\put(2654,-1525){\line( 1, 0){517}}

\put(691,217){$\left| {0 } \right\rangle$}
\put(691,-46){$\left| {0 } \right\rangle$}
\put(691,-396){$\left| {0 } \right\rangle$}
\put(691,-752){$\left| {1 } \right\rangle$}
\put(691,-1003){$\left| {1 } \right\rangle$}
\put(691,-1318){$\left| {1 } \right\rangle$}
\put(691,-1601){$\left| {0 } \right\rangle$}

\put(3233,217){$\left| {x_0 } \right\rangle$}
\put(3233,-46){$\left| {x_1 } \right\rangle$}
\put(3233,-396){$\left| {x_{n-1} } \right\rangle$}
\put(3233,-752){$\left| {c_0 } \right\rangle$}
\put(3233,-1003){$\left| {c_1 } \right\rangle$}
\put(3233,-1318){$\left| {c_{m-1} } \right\rangle$}
\put(3233,-1601){$\left| {ax } \right\rangle$}

\put(691,-1178){$\vdots$}
\put(956,-1178){$\vdots$}
\put(1670,-1178){$\vdots$}
\put(3233,-1178){$\vdots$}

\put(691,-243){$\vdots$}
\put(956,-243){$\vdots$}
\put(3233,-243){$\vdots$}

\put(2115,-1802){$O(m^2)$}
\put(2175,-1195){$M_x$}

\put(1190,-396){$H$}
\put(1190,-46){$H$}
\put(1190,217){$H$}

\put(1601,-163){$C_f$}
\end{picture}%

\end{center}
\caption{A quantum circuit for the proposed algorithm.}
\label{alg}
\end{figure*}
\end{center}

\begin{itemize}
\item[1-]{Variable Assignments Preparation.} To prepare a superposition of all variable assignments of $n$ qubits, 
apply $H^{\otimes n}\otimes I^{\otimes m+1}$ on the $n+m+1$ qubits

\begin{equation}
\begin{array}{l}
 \left| {\psi _1 } \right\rangle  = \left( {H^{ \otimes n}  \otimes I^{ \otimes m + 1} } \right)\left| {\psi _0 } \right\rangle  \\ 
 \,\,\,\,\,\,\,\,\,\,\,\,  = \frac{1}{{\sqrt N }}\sum\limits_{k = 0}^{N - 1} {\left| {A_k } \right\rangle }  \otimes \left| 1 \right\rangle ^m  \otimes \left| 0 \right\rangle, \\ 
 \end{array}
\end{equation}

\noindent
where $H$ is the Hadamard gate, $I$ is the identity matrix of size $2 \times 2$, $N=2^n$, and 
$A_k  = \left( {x_0^k ,x_1^k , \ldots ,x_{n - 1}^k } \right) \in \left\{ {0,1} \right\}^n$ is the bit-wise 
representation of an integer $k$, for $0\le k \le N-1$, that represents a variable 
assignment out of the $N$ possible variable assignments. 

 \item[2-]{Preparation of the Truth Vectors of Clauses}. For every E3-CNF clause $c_j=(l_0 \vee l_1 \vee l_2)$, 
 apply a $GT^4$ gate taking qubit $j$ in the $m$ qubits register as the target qubit as shown in Section \ref{cfsec}. 
 The collection of all $GT^4$ gates applied to evaluate the $m$ clauses is denoted $C_f$ in figure \ref{alg}, 
then the system is transformed to,

\begin{equation}
\begin{array}{l}
 \left| {\psi _2 } \right\rangle  = \left( {C_f  \otimes I} \right)\left| {\psi _1 } \right\rangle  \\ 
 \,\,\,\,\,\,\,\,\,\,\,\,= \frac{1}{{\sqrt N }}\sum\limits_{k = 0}^{N - 1} {\left( {\left| {A_k } \right\rangle  \otimes \left| {C_k } \right\rangle } \right)}  \otimes \left| 0 \right\rangle,  \\ 
 \end{array}
\end{equation}
\noindent
where $C_k  = \left( {c_0^k ,c_1^k , \ldots ,c_{m - 1}^k } \right) \in \left\{ {0,1} \right\}^m$ is the truth vector 
for the $m$ clauses associated with variable assignment $A_k$.

\begin{center}
\begin{figure*}[t]
\begin{center}

\setlength{\unitlength}{3947sp}%
\begingroup\makeatletter\ifx\SetFigFont\undefined
\def\x#1#2#3#4#5#6#7\relax{\def\x{#1#2#3#4#5#6}}%
\expandafter\x\fmtname xxxxxx\relax \def\y{splain}%
\ifx\x\y   
\gdef\SetFigFont#1#2#3{%
  \ifnum #1<17\tiny\else \ifnum #1<20\small\else
  \ifnum #1<24\normalsize\else \ifnum #1<29\large\else
  \ifnum #1<34\Large\else \ifnum #1<41\LARGE\else
     \huge\fi\fi\fi\fi\fi\fi
  \csname #3\endcsname}%
\else
\gdef\SetFigFont#1#2#3{\begingroup
  \count@#1\relax \ifnum 25<\count@\count@25\fi
  \def\x{\endgroup\@setsize\SetFigFont{#2pt}}%
  \expandafter\x
    \csname \romannumeral\the\count@ pt\expandafter\endcsname
    \csname @\romannumeral\the\count@ pt\endcsname
  \csname #3\endcsname}%
\fi
\fi\endgroup
\begin{picture}(2771,2237)(585,-1718)
\thinlines
\put(3002,-1494){\oval(84,92)}
\put(3001,-1539){\line( 0, 1){ 94}}
\put(2678,-1514){\oval(236, 72)[tr]}
\put(2678,-1514){\oval(236, 72)[tl]}
\put(1646,-939){\circle*{90}}
\put(1278,-715){\circle*{90}}
\put(2206,-1270){\circle*{90}}

\put(1186,-1610){\framebox(176,215){}}
\put(1554,-1610){\framebox(176,215){}}
\put(1276,-686){\line( 0,-1){705}}
\put(943,-1505){\line( 1, 0){235}}
\put(1364,-1501){\line( 1, 0){188}}
\put(2112,-1614){\framebox(176,215){}}
\put(2288,-1501){\line( 1, 0){188}}
\put(2016,-935){\line( 1, 0){1328}}
\put(2016,-1277){\line( 1, 0){1328}}
\put(1824,-1273){\line(-1, 0){893}}
\put(1824,-935){\line(-1, 0){893}}
\put(1828,-716){\line(-1, 0){893}}
\put(2015,-1501){\line( 1, 0){ 94}}
\put(1737,-1499){\line( 1, 0){ 94}}
\put(2202,-1255){\line( 0,-1){141}}
\put(1644,-927){\line( 0,-1){470}}
\put(1098,-1706){\framebox(2118,1155){}}
\put(2482,-1397){\line( 0,-1){235}}
\put(2482,-1632){\line( 1, 0){376}}
\put(2858,-1632){\line( 0, 1){235}}
\put(2858,-1397){\line(-1, 0){376}}
\put(2567,-1572){\vector( 2, 1){282}}
\put(2866,-1496){\line( 1, 0){470}}
\put(2012,-716){\line( 1, 0){1316}}

\put(480,-1558){$\left| 0 \right\rangle$}
\put(480,-714){$\left| c_0 \right\rangle$}
\put(480,-989){$\left| c_1 \right\rangle$}
\put(480,-1319){$\left| c_{m-1} \right\rangle$}

\put(3385,-714){$\left| c_0(A_{max}) \right\rangle$}
\put(3385,-989){$\left| c_1(A_{max}) \right\rangle$}
\put(3385,-1319){$\left| c_{m-1}(A_{max}) \right\rangle$}
\put(3385,-1558){$\left| ax \right\rangle$}

\put(956,-1180){$\vdots$}

\put(1215,-1584){$V$}
\put(1578,-1579){$V$}
\put(2139,-1593){$V$}

\put(1845,-1319){$\ldots$}
\put(1845,-989){$\ldots$}
\put(1845,-714){$\ldots$}
\put(1845,-1558){$\ldots$}

\put(1750,-475){$M_x$}
\put(1750,-1880){$O(m^2)$}
\end{picture}%

\end{center}
\caption{Quantum circuits for the $M_x$ operator followed by a partial measurement then a negation 
to reset the auxiliary qubit $\left| {ax} \right\rangle$. }
\label{mmfig}
\end{figure*}
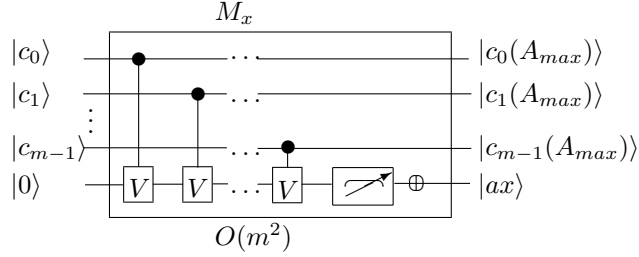
\end{center}

\item[3-]{Maximization of the Number of Satisfied Clauses}. 
The aim of this stage is to find the state $\left| C_k \right\rangle$ that contains the maximum number of 
$\left| 1 \right\rangle$s. Such a state will be denoted $\left| C_{max} \right\rangle$. A modified version 
of the amplitude amplification algorithm shown in \cite{MaxMinBis} 
will be used for this purpose. Every $\left| C_k \right\rangle$ is entangled with 
the corresponding variable assignment $\left| A_k \right\rangle$. Since the variable assignment 
$\left| A_k \right\rangle$ will not be involved directly in this stage and the corresponding 
$\left| C_k \right\rangle$ will not be modified by any operation, then for simplicity the system can be re-written as,

\begin{equation}
 \left| {\psi _3 } \right\rangle  = \frac{1}{{\sqrt N }}\sum\limits_{k = 0}^{N - 1} {\left| {C_k } \right\rangle }  \otimes \left| 0 \right\rangle.
\end{equation}
  
Let $d_k=\left\langle {C_k } \right\rangle$ be the 1-density of the state  $\left| {C_k } \right\rangle$. 
A solution to the MAX-E3-SAT problem is to find the state 
$\left| {C_{max}} \right\rangle$ with $d_{max}=max\{d_k,\,0\le k \le N-1\}$.

The aim is to find $\left| {C_{max} } \right\rangle$ when $\left| {\psi _3 } \right\rangle$ is measured. 
To find $\left| {C_{max} } \right\rangle$, the algorithm applies partial negation on 
the state of $\left| {ax} \right\rangle$ entangled with $\left| {C_k } \right\rangle$ based on 
the 1-density of $\left| {C_k } \right\rangle$, i.e. more 1's in $\left| {C_k } \right\rangle$ 
gives more negation to the state of $\left| {ax} \right\rangle$ entangled with $\left| {C_k } \right\rangle$. 
If the number of 1's in $\left| {C_k } \right\rangle$ is $m$, then 
the entangled state of $\left| {ax} \right\rangle$ will be fully negated.

Let $X$ be the Pauli-X gate which is the quantum equivalent to 
the NOT gate. It can be seen as a rotation of the Bloch Sphere around the X-axis by $\pi$ radians as follows,

\begin{equation}
X  = \left[ {\begin{array}{*{20}c}
   0 & 1  \\
   1 & 0  \\
\end{array}} \right].
\end{equation}

The $m^{th}$ partial negation operator $V$ is the $m^{th}$ root of the $X$ gate and 
can be calculated using diagonalization as follows, 

\begin{equation}
 V=\sqrt[m]{X} = \frac{1}{2}\left[ {\begin{array}{*{20}c}
   {1 + t} & {1 - t}  \\
   {1 - t} & {1 + t}  \\
\end{array}} \right],
\end{equation}

\noindent
where $t={\sqrt[m]{{ - 1}}}$, and applying $V$ for $d$ times on a qubit is equivalent to the operator,
 
\begin{equation}
 V^d  = \frac{1}{2}\left[ {\begin{array}{*{20}c}
   {1 + t^d } & {1 - t^d }  \\
   {1 - t^d } & {1 + t^d }  \\
\end{array}} \right],
\end{equation}

\noindent
such that if $d=m$, then $V^m=X$. To amplify the amplitude of the state $\left| {C_{max} } \right\rangle$, 
apply the operator $M_x$ 
on $\left| {\psi _3 } \right\rangle$ as will be shown later, 
where $M_x$ is an operator on $m+1$ qubits register that applies $V$ conditionally 
for $m$ times on $\left|ax \right\rangle$ 
based on 1-density of $\left| {c_0 c_1  \ldots c_{m-1} } \right\rangle$ 
as follows (as shown in figure \ref{mmfig}),

\begin{equation}
M_x = Cont\_V(c_0 ;ax )Cont\_V(c_1 ;ax ) \ldots Cont\_V(c_{m - 1} ;ax),
\end{equation}
\noindent
where the $Cont\_V(c_j ;ax )$ gate is a 2-qubit controlled gate with control qubit ${\left| c_j \right\rangle }$ and 
target qubit ${\left| ax \right\rangle }$. The $Cont\_V(c_j ;ax )$ gate applies $V$ conditionally on 
${\left| ax \right\rangle }$ if ${\left| c_j \right\rangle }= {\left| 1 \right\rangle }$, 
so, if $d$ is the 1-density of $\left| {c_0 c_1  \ldots c_{m-1} } \right\rangle$ then,

\begin{equation}
M_x\left( {\left| {c_0 c_1 ...c_{m - 1} } \right\rangle  \otimes \left| 0 \right\rangle } \right) = \left| {c_0 c_1 ...c_{m - 1} } \right\rangle  \otimes \left( {\frac{{1 + t^{d} }}{2}\left| 0 \right\rangle  + \frac{{1 - t^{d} }}{2}\left| 1 \right\rangle } \right),
\end{equation}
\noindent
and the probabilities of finding the auxiliary qubit $\left|ax \right\rangle$ in state 
${\left| 0 \right\rangle }$ or ${\left| 1 \right\rangle }$ when measured is respectively as follows,

\begin{equation}
\begin{array}{l}
 Pr{(\left| ax \right\rangle = \left| 0 \right\rangle)}  = \left| {\frac{{1 + t^d }}{2}} \right|^2  = \cos ^2 \left( {\frac{{d\pi }}{{2m}}} \right), \\ 
 Pr{(\left| ax \right\rangle = \left| 1 \right\rangle)}  = \left| {\frac{{1 - t^d }}{2}} \right|^2  = \sin ^2 \left( {\frac{{d\pi }}{{2m}}} \right). \\ 
 \end{array}
\end{equation}


To find the state ${\left| C_{max} \right\rangle }$ in $\left| {\psi _3 } \right\rangle$, 
the proposed algorithm is as shown in Algorithm 1  and as shown in figure \ref{mmfig}. 
For simplicity and without loss of generality, assume that a single $\left| C_{max} \right\rangle$ 
exists in $\left| \psi_3 \right\rangle$, although different variable assignments might be associated 
with truth vectors with maximum 1-density with different order of 1's, but such information is not known in advance.


\begin{algorithm}
\label{alg1}
\caption{Amplify $\left| C_{max} \right\rangle$ in  $\left| {\psi _3 } \right\rangle$ }
\begin{algorithmic}[1] 
\STATE Let $\left| {\psi _r } \right\rangle = \left| {\psi _3 } \right\rangle$ 
\FOR {$counter = 1 \to r$}
	\STATE Apply the operator $M_x$ on $\left| {\psi _r } \right\rangle$.
	\STATE Measure $\left| ax \right\rangle$
	\IF {$\left|ax \right\rangle=\left|1 \right\rangle$}
		\STATE Let $\left| {\psi _r } \right\rangle$ be the system post-measurement of $\left|ax \right\rangle$ 
		\STATE Apply $X$ gate on $\left| ax \right\rangle$ 
		\COMMENT {to reset $\left| ax \right\rangle$ to $\left| 0 \right\rangle$ for the next iteration} 
	\ELSE
		\STATE  Let $\left| {\psi _r } \right\rangle = \left| {\psi _3 } \right\rangle$  and restart the for-loop
\ENDIF
\ENDFOR
\STATE Measure the first $m$ qubits in $\left| {\psi _r } \right\rangle$ to read $\left| C_{max} \right\rangle$.

\IF {$\left| C_{max} \right\rangle=\left| {1} \right\rangle ^{\otimes m}$}
	\STATE The E3-CNF formula is satisfiable
\ELSE
	\STATE The E3-CNF formula is not satisfiable where number of $\left| {1} \right\rangle$'s in $\left| C_{max} \right\rangle$ represents the maximum number of satisfied clauses in order
\ENDIF
\STATE Measure the first $n$ qubits in $\left| {\psi _2 } \right\rangle$ to read the corresponding variable assignment $\left| A_{max} \right\rangle$

\end{algorithmic}
\end{algorithm}

We require that Algorithm 1 finds $\left|ax \right\rangle=\left|1 \right\rangle$ for $r$ times 
in a row. The probability of finding $\left|ax \right\rangle=\left|1 \right\rangle$ after Line:4 in 
the $1^{st}$ iteration of the for-loop is given by,

\begin{equation}
Pr^{(1)}{(\left| ax \right\rangle = \left| 1 \right\rangle)}  = \frac{1}{N} \sum\limits_{k = 0}^{N - 1} {\sin ^2 \left( {\frac{{d_k \pi }}{{2m}}} \right)}.
\label{probax2}
\end{equation} 

The probability of finding $\left|\psi_r \right\rangle=\left|C_{max} \right\rangle$ after Line:4 in 
the $1^{st}$ iteration, i.e. $r=1$ is given by,

\begin{equation}
Pr^{(1)}{(\left|\psi_r \right\rangle=\left|C_{max} \right\rangle)}  = \frac{1}{N} {\sin ^2 \left( {\frac{{d_{max} \pi }}{{2m}}} \right)} .
\end{equation}

The probability of finding $\left|ax \right\rangle=\left|1 \right\rangle$ after Line:4 in 
the $r^{th}$ iteration, is given by,

\begin{equation}
Pr^{(r)}{(\left| ax \right\rangle = \left| 1 \right\rangle)}  = \frac{{\sum\limits_{k = 0}^{N - 1} {\sin ^{2r} \left( {\frac{{d_k \pi }}{{2m}}} \right)} }}{{\sum\limits_{k = 0}^{N - 1} {\sin ^{2(r - 1)} \left( {\frac{{d_k \pi }}{{2m}}} \right)} }}.
\end{equation}

\begin{center}
\begin{figure*}[t]
\begin{center}
   \vspace{20pt}%
   \includegraphics[width=250pt]{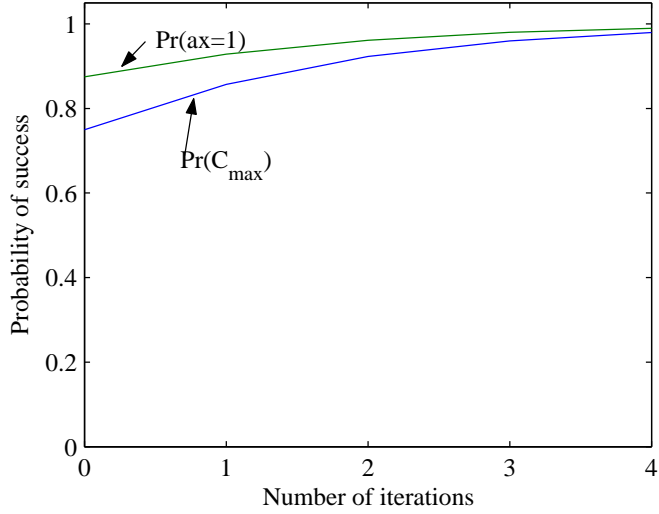}
   \caption{The probability of success for an E3-CNF formula: $\left( {x_0  \vee x_1  \vee x_2 } \right) \wedge \left( {\neg x_0  \vee \neg x_1  \vee \neg x_2 } \right)$  
   with $n=3$ and $m=2$ with $d_{max}=2$, i.e. the maximum number of satisfied clauses is 2, 
   where $Pr^{(1)}{(\left|\psi_{r}\right\rangle  = \left|C_{max}\right\rangle)}=0.75$, $Pr^{(1)}{(\left|ax\right\rangle  = \left|1\right\rangle)}=0.875$, $Pr^{(r)}{(\left|\psi_{r}\right\rangle  = \left| C_{max}\right\rangle)}=0.98$, 
and $Pr^{(r)}{(\left|ax\right\rangle  = \left|1\right\rangle)}=0.99$.}
   \label{fig21}
\end{center}
\end{figure*}
\end{center}

The probability of finding $\left|\psi_r \right\rangle=\left|C_{max} \right\rangle$ after Line:4 in the $r^{th}$ iteration, i.e. $r>1$ is given by,

\begin{equation}
Pr^{(r)}{(\left|\psi_r \right\rangle=\left|C_{max} \right\rangle)}  = \frac{{{\sin ^{2r} \left( {\frac{{d_{max} \pi }}{{2m}}} \right)} }}{{\sum\limits_{k = 0}^{N - 1} {\sin ^{2(r - 1)} \left( {\frac{{d_k \pi }}{{2m}}} \right)} }}.
\end{equation} 

To get the highest probability of success for $Pr{(\left|\psi_r \right\rangle=\left|C_{max} \right\rangle)}$, the for-loop should be repeated until 
$\left| Pr^{(r)}{(\left| ax \right\rangle = \left| 1 \right\rangle)}  - Pr^{(r)}{(\left|\psi_r \right\rangle=\left|C_{max} \right\rangle)} \right| \le \epsilon$ for small $\epsilon \ge 0$ 
as shown in figure \ref{fig21}. 
This happens when,
 
\begin{equation}
\sum\limits_{k = 0,k\ne max}^{N - 1} {\sin ^{2r} \left( {{\textstyle{{d_k \pi } \over {2m}}}} \right)}  \le \epsilon, 
\end{equation}
\noindent
and since the Sine function is a decreasing function then for sufficient large $r$,

\begin{equation}
\sum\limits_{k = 0,k\ne max}^{N - 1} {\sin ^{2r} \left( {\frac{{d_k \pi }}{{2m}}} \right)}  \approx \sin ^{2r} \left( {\frac{{d_{nm} \pi }}{{2m}}} \right),
\end{equation} 
\noindent
where $d_{nm}$ is the next maximum 1-density less than $d_{max}$. The values of $d_{max}$ and $d_{nm}$ are unknown in advance, 
so let $d_{max}=m$ be the number of satisfied clauses, 
then in the worst case when $d_{max}=m$, $d_{nm}=m-1$ and $m={\textstyle{4 \over 3}}n(n - 1)(n - 2)$, 
the required number of iterations $r$ 
for $\epsilon  = 10^{ - \lambda }$ and $\lambda>0$ can be calculated using the formula,

\begin{equation}
0 < \sin ^{2r} \left( {\frac{{(m-1) \pi }}{{2m}}} \right) \le \epsilon,
\end{equation} 
\noindent
then,

\begin{equation}
\begin{array}{l}
 r \ge \frac{{\log \left( \epsilon  \right)}}{{2\log \left( {\sin \left( {\frac{{\left( {m - 1} \right)\pi }}{{2m}}} \right)} \right)}} \\ 
 \,\,\,\, = \frac{{\log \left( {10^{ - \lambda } } \right)}}{{2\log \left( {\cos \left( {\frac{\pi }{{2m}}} \right)} \right)}} \\ 
 \,\,\,\, \ge \lambda \left( {\frac{{2m}}{\pi }} \right)^2 =O\left( {m^2 } \right), \\
 \end{array}
\end{equation} 
\noindent
where $0 < m \le {\textstyle{4 \over 3}}n(n - 1)(n - 2)$. 
When $m={\textstyle{4 \over 3}}n(n - 1)(n - 2)$, then the upper bound 
for the required number of iterations $r$ is $O\left( {n^6 } \right)$. Assuming that 
a single $\left|C_{max} \right\rangle$ exists in the superposition will increase the required number of iterations, 
so it is important to notice here that the probability of success will not be over-cooked by increasing the 
required number of iteration $r$ similar to the common amplitude 
amplification techniques.

\end{itemize}

\subsection{Tuning the Probability of Success}

During the above analysis, two problems might arise during the implementation of the proposed algorithm.
The first one is to finding $\left|ax \right\rangle=\left|1 \right\rangle$ for $r$ times in a row 
which is a critical issue in the 
success of the proposed algorithm to terminate in polynomial time. The second problem is that the value of $d_{max}$ 
is not known in advance, where the value of $Pr^{(1)}{(\left|ax\right\rangle  = \left|1\right\rangle)}$ 
shown in equation (\ref{probax2}) plays an important role in the success of finding 
$\left|ax \right\rangle=\left|1 \right\rangle$ in the next iterations, this value depends heavily on 
the 1-density of $\left| C_{max} \right\rangle$, i.e. the ratio ${\textstyle{{d_{max} } \over m}}$.

Consider the case of a complete E3-CNF formula where the number of clauses is 
$m={\textstyle{4 \over 3}}n(n - 1)(n - 2)$ and all the $\left|C_k \right\rangle$'s are equivalent 
where anyone can be taken as $\left|C_{max} \right\rangle$. In this case, each clause $c_j$ will be satisfied by 
7 variable assignments out of 8 possible variable assignment, then $d_{max}=\frac{7}{8} m$ for any $\left|C_k \right\rangle$ 
\cite{Trevisan:1996}, so $Pr^{(1)}{(\left|ax\right\rangle  = \left|1\right\rangle)}$ is as follows, 

\begin{equation}
\begin{array}{l}
Pr^{(1)}{(\left|ax\right\rangle  = \left|1\right\rangle)}  =  {\sin ^2 \left( {\frac{{d_{max}\pi }}{{2m}}} \right)}\\
\,\,\,\,\,\,\,\,\,\,\,\,\,\,\,\,\,\,\,\,\,\,\,\,\,\,\,\,\,\,\,\,\,\,\,\,\,\,\,\,\,\,\,\,\, = {\sin ^2 \left( {\frac{{7\pi }}{{16}}} \right)} \\
\,\,\,\,\,\,\,\,\,\,\,\,\,\,\,\,\,\,\,\,\,\,\,\,\,\,\,\,\,\,\,\,\,\,\,\,\,\,\,\,\,\,\,\,\, = 0.9619.
\end{array}
\label{probax2_2}
\end{equation}

This case is a trivial case for the proposed algorithm by setting $m=d_{max}$ in $m^{th}$ root of $X$ to get a probability of 
success of certainty after a single iteration. Assuming a blind approach where $d_{max}$ is not known, 
then this case represents the worst case \cite{Zhang:2001} and iterating the proposed algorithm will not amplify 
the amplitudes after arbitrary number of iterations. 
For an arbitrary E3-CNF formula, the actual probability of success will depend of 
the 1-density of $\left|C_{max} \right\rangle$, i.e. the ratio ${\textstyle{{d_{max} } \over m}}$. 
In the following, a tuning of $Pr^{(1)}{(\left|ax\right\rangle  = \left|1\right\rangle)}$ 
will be shown so that we can find $\left|ax \right\rangle=\left|1 \right\rangle$ 
after the first iteration with an arbitrary higher probability of success close to certainty 
without a priori knowledge of $d_{max}$.

\begin{center}
\begin{figure*}[t]
\begin{center}
   \vspace{20pt}%
   \includegraphics[width=250pt]{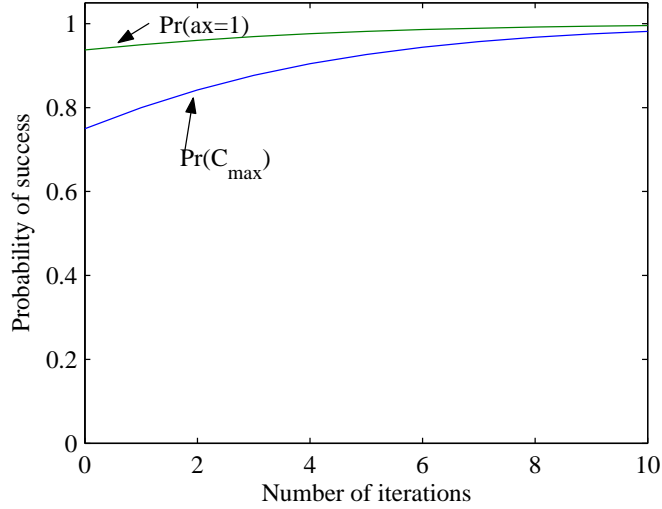}
      \caption{The probability of success for the E3-CNF formula shown in figure \ref{fig21} 
where $Pr^{(1)}{(\left|ax\right\rangle  = \left|1\right\rangle)}$ is raised from 0.875 to 0.94 by adding a single 
temporary qubit initialized to state $\left|1 \right\rangle$, i.e. $\mu_{max}=1$ 
where $Pr^{(1)}{(\left|\psi_{r}\right\rangle  = \left|C_{max}\right\rangle)}=0.75$, $Pr^{(1)}{(\left|ax\right\rangle  = \left|1\right\rangle)}=0.94$, $Pr^{(r)}{(\left|\psi_{r}\right\rangle  = \left|C_{max}\right\rangle)}=0.98$, 
and $Pr^{(r)}{(\left|ax\right\rangle  = \left|1\right\rangle)}=0.99$.}
   \label{fig22}
\end{center}
\end{figure*}
\end{center}

For an arbitrary E3-CNF formula, we could interpret the formula for $Pr(\left| ax \right\rangle = \left| 1 \right\rangle)$ in equation (\ref{probax2}) 
as the expected value of the function, 
\begin{equation}
\phi(x) = sin^2 \left( \frac {x \pi} {2}\right),
\end{equation}
\noindent
where $x$ is the proportion of clauses satisfied by a random truth assignment, that is,  $Pr(\left| ax \right\rangle = \left| 1 \right\rangle) = E[ \phi(x) ]$. 
The bounds for the probability of finding $\left|ax\right\rangle = \left|1\right\rangle$ in the first iteration is 
as shown in the following Lemma,

\begin{lemma}
The probability of finding $\left|ax\right\rangle = \left|1\right\rangle$ in the first iteration is bounded as follows,
\begin{equation}
0.691 < 1 - \frac{\pi^2}{32} \leq Pr^{(1)}(\left| ax \right\rangle = \left| 1 \right\rangle)  \leq \sin( \frac {7 \pi}{16})  < 0.981.
\end{equation}
\end{lemma}

\begin{proof}
\begin{equation}
Pr ^{(1)} (\left| ax \right\rangle = \left| 1 \right\rangle) = E\left[ {\sin ^2 \left( {\frac{{d_k \pi }}{{2m}}} \right)} \right] \le E\left[ {\sin \left( {\frac{{d_k \pi }}{{2m}}} \right)} \right].
\end{equation}

Since $\sin(x)$ is a concave function on $0 \le x \le \pi/2$, it follows from Jensen's inequality that,
\begin{equation}
  Pr^{(1)}(\left| ax \right\rangle = \left| 1 \right\rangle)  \le \sin \left(\frac{ E[d_k] \pi }{ 2 m } \right) =  \sin\left( \frac{7 \pi}{16} \right)  < 0.981.
\end{equation}

For the lower bound we use, 
\begin{equation}
    \sin^2 ( (1 - x) \pi/2 ) \geq 1 - \frac{\pi^2 x^2}{4},
\end{equation}
which follows from the Taylor series taken around $x = 0$. Then,
\begin{equation}
\begin{array}{l}
Pr^{(1)}(\left| ax \right\rangle = \left| 1 \right\rangle)  \ge \frac{1}{N} \sum_{k=0}^{N-1} \left( 1 - \frac{\pi^2 (1 - d_k/m)^2}{4} \right)\\
\,\,\,\,\,\,\,\,\,\,\,\, \,\,\,\,\,\,\,\,\,\,\,\,\,\,\,\,\,\,\,\,\,\,\,\,\,\,= 1 + \frac{3 \pi^2}{16} - \frac{\pi^2}{4N m^2} \sum_{k=0}^{N-1} d_k^2 \\
\,\,\,\,\,\,\,\,\,\,\,\, \,\,\,\,\,\,\,\,\,\,\,\,\,\,\,\,\,\,\,\,\,\,\,\,\,\,\ge 1 + \frac{3 \pi^2}{16} - \frac{\pi^2}{4N m^2} \sum_{k=0}^{N-1} d_k m\\
\,\,\,\,\,\,\,\,\,\,\,\, \,\,\,\,\,\,\,\,\,\,\,\,\,\,\,\,\,\,\,\,\,\,\,\,\,\,= 1 + \frac{3 \pi^2}{16} - \frac{\pi^2}{4N m} \sum_{k=0}^{N-1} d_k\\
\,\,\,\,\,\,\,\,\,\,\,\, \,\,\,\,\,\,\,\,\,\,\,\,\,\,\,\,\,\,\,\,\,\,\,\,\,\,= 1 + \frac{3 \pi^2}{16} - \frac{7 \pi^2}{32}\\
\,\,\,\,\,\,\,\,\,\,\,\, \,\,\,\,\,\,\,\,\,\,\,\,\,\,\,\,\,\,\,\,\,\,\,\,\,\,= 1 - \frac{\pi^2}{32}  > 0.691.\\
\end{array}
\end{equation}

\end{proof}

To overcome the problem of low probability of finding $\left|ax\right\rangle = \left|1\right\rangle$ in the first iteration, 
we can add $\mu_{max}$ temporary qubits initialized to state $\left|1 \right\rangle$ 
to the register $\left| {C_k } \right\rangle$ as follows,

\begin{equation}
\left| {c_0 c_1 ...c_{m - 1} } \right\rangle \to \left| {c_0 c_1 \ldots c_{m - 1}c_{m}c_{m+1}\ldots c_{m+\mu_{max}-1} } \right\rangle,
\end{equation}
\noindent
so that the extended number of clauses $m_{ext}$ will be $m_{ext}=m+\mu_{max}$ and $V=\sqrt[{m_{ext} }]{{X }}$ will be used 
instead of  $V=\sqrt[{m}]{{ X }}$ in the $M_x$ operator, then the density of 1's will be 
$\frac{{{\textstyle{7 \over 8}}m + \mu _{\max } }}{{m + \mu _{\max } }}$. To get a probability of success $Pr_{max}$ 
to find $\left|ax \right\rangle=\left|1 \right\rangle$ after the first iteration of the for-loop in Algorithm 1,

\begin{equation}
Pr ^{(1)} \left( {\left|ax\right\rangle  = \left|1\right\rangle} \right) = N\alpha ^2 \sin ^2 \left( {\frac{{\pi \left( {{\textstyle{7 \over 8}}m + \mu _{\max } } \right)}}{{2\left( {m + \mu _{\max } } \right)}}} \right) \ge Pr_{\max },
\end{equation}
\noindent
then the required number of temporary qubits $\mu_{max}$ is calculated as follows,

\begin{equation}
\mu _{\max }  \ge m\left( {\frac{{\omega  - \frac{7}{8}}}{{1 - \omega }}} \right),
\end{equation}
\noindent
where $\omega  = {\textstyle{2 \over \pi }}\sin ^{ - 1} \left( {\sqrt {{{{Pr_{\max }} }}} } \right)$ 
and $N\alpha ^2 =1$. For example, if $Pr_{\max }=0.99$, 
then $Pr ^{(1)} \left( \left| ax \right\rangle = \left| 1 \right\rangle \right)$ will be in the neighborhood of 99$\%$ as shown in figure \ref{fig22}. To conclude, 
the problem of low 1-density of $\left| {C_{max} } \right\rangle$ can be solved with a polynomial increase in the number of qubits to get the solution 
$\left|C_{max} \right\rangle$ in $O\left(m_{ext}^2\right)=O\left( {n^6 } \right)$ iterations with arbitrary 
high probability $Pr_{max}<1$ to terminate in poly-time, i.e. to read $\left|ax \right\rangle=\left|1 \right\rangle$ 
for $r$ times in a row.



\section{Conclusion}

Given an E3-CNF Boolean formula with $n$ inputs, the paper showed that BQP contains NP 
in a non-relativized world by proposing a BQP quantum algorithm to solve the MAX-E3-SAT problem with $m$ clauses. 
The proposed algorithm encoded every clause as a GT$^4$ gate where $O(n+m)$ qubits are used.
The algorithm is divided into three stages; the first stage prepares a superposition of all possible 
variable assignments. In the second stage, the algorithm evaluates the set of clauses for all possible variable assignments 
using a quantum circuit composed of GT$^4$ gates so that each variables assignment is entangled with a truth vector 
of clauses evaluated according to that variables assignment. In the third stage, the algorithm amplified 
the amplitudes of the truth vector of clauses that achieves the maximum satisfaction to the set of clauses 
using an amplitude amplification technique that applies an iterative partial negation 
where partial negation is applied to the state of an auxiliary qubit entangled with the truth vector of clauses based on the number 
of satisfied clauses, i.e. more satisfied clauses implies more negation to entangled state of the auxiliary qubit. 
A partial measurement on the auxiliary qubit is then used to amplify the set of clauses with more negation. 
The third stage requires $O(m^2)$ iterations and in the worst case requires 
$O(n^6)$ iterations. It was shown that the proposed algorithm achieves an arbitrary high probability of success of 
$1-\epsilon$ for small $\epsilon>0$ using a polynomial increase in the resources by adding dummy clauses with predefined 
values to give more negation to the best truth vector of clauses. 

In the same manner, the proposed algorithm can also be used decide if a given E3-CNF Boolean formula is satisfiable or not 
by checking the truth vector of clauses, if the $m$ clauses are satisfied then the E3-CNF Boolean formula is satisfiable, 
if not, then the proposed algorithm gives the maximum number of satisfied clauses with the corresponding variable assignment.

The proposed algorithm can easily be extended in a trivial way to solve/decide an E$k$-CNF Boolean formula by encoding 
any E$k$-CNF clause as a GT$^{k+1}$ gate where it can be shown that the algorithm will require $O(n^{2k})$ iterations.

\end{document}